\documentclass[letter, 10pt, conference]{IEEEtran} 

\usepackage{lmodern}
\usepackage{amsmath,amsthm}
\usepackage{amssymb}

\usepackage{mathtools}

\usepackage{enumitem}
\setenumerate{nolistsep}
\setitemize{nolistsep}

\allowdisplaybreaks
\usepackage{defs}
\DeclareMathOperator*{\ec}{ec}

\DeclareMathOperator{\diag}{diag}
\newcommand{\VA}{V_{A}}
\newcommand{\VB}{V_{B}}
\newcommand{\EA}{E_A}
\newcommand{\EB}{E_B}
\newcommand{\GA}{G(A)}
\newcommand{\HA}{H(A)}
\newcommand{\Ze}{\mathcal{Z}}
\newcommand{\F}{\mathcal{F}}

\newcommand{\NS}{N^{-}}
\newcommand{\rt}{R}
\newcommand{\Xe}{X_e}
\newcommand{\Se}{S_e}

\newcommand{\rtn}{\tilde{R}}

\newcommand{\de}{d^{-}}
\newcommand{\dep}{d^{+}}
\newcommand{\I}{\mathcal{I}}
\newcommand{\GAB}{G(A, B)}
\newcommand{\bil}{V_{A}^{1}}
\newcommand{\bir}{V_{A}^{2}}
\newcommand{\match}{M}
\newcommand{\gra}{G}

\newcommand{\sys}{A}
\newcommand{\inp}{B}

\title{Resilience of Complex Networks}

\author{Priyanka Dey, Niranjan Balachandran, and Debasish Chatterjee%
	\thanks{PD and DC are with Systems \& Control Engineering, and NB is with the Department of Mathematics, IIT Bombay, Powai, Mumbai 400076, India. Emails: \textsf{dey\_priyanka@sc.iitb.ac.in, niranj@math.iitb.ac.in, dchatter@iitb.ac.in}}%
}

\begin{document}

    \maketitle

	\begin{abstract}
	This article determines and characterizes the minimal number of actuators needed to ensure structural controllability of a linear system under structural alterations that can severe the connection between any two states. We assume that initially the system is structurally controllable with respect to a given set of controls, and propose an efficient system-synthesis mechanism to find the minimal number of additional actuators required for resilience of the system w.r.t such structural changes. The effectiveness of this approach is demonstrated by using standard IEEE power networks.
	
	\end{abstract}

\section{Introduction}
\label{s:intro}
Consider a controlled linear system
\begin{equation}
	\label{e:linsys}
	\dot{\st}(t)= \sys \st(t) + \inp\cont(t),
\end{equation}
where \(\st(t) \in \R^{d}\) are the states and \(\cont(t) \in \R^{m}\) are the control actions at time \(t\), and \(\sys \in \R^{d \times d}\) and \(\inp \in \R^{d \times m}\) are the given state and control matrices, respectively. We assume that \eqref{e:linsys} is controllable,\footnote{For us controllability of \eqref{e:linsys} is equivalent to the rank of the matrix \(\pmat{B & AB & \cdots & A^{d-1} B}\) being \(d\).} and let it under structural alterations. These alterations may be of different kinds, e.g., 
\begin{enumerate}[label=(\roman*), leftmargin=*, widest=iii, align=right]
	\item \label{alter:states} connections between certain states may be severed,
	\item \label{alter:inputs-to-states} connections between some of the control inputs and some of the states may be severed, or
	\item \label{alter:inputs} some of the control inputs may become dysfunctional.
\end{enumerate}
Structural changes of the type \ref{alter:states} reflect in the system \eqref{e:linsys} in the form of certain non-zero entries of \(A\) being set to \(0\), \ref{alter:inputs-to-states} reflect as certain elements of the control matrix \(B\) being set to \(0\), and \ref{alter:inputs} reflect as certain columns of \(B\) being set to \(0\).

The structural alterations \ref{alter:states}-\ref{alter:inputs-to-states}-\ref{alter:inputs} above may be consequences of natural causes such as ageing/malfunctioning of system components \cite{ref:AtuSah-09}, or due to malicious external attacks that are designed to adversely affect normal operations of the systems \cite{ref:PasDorBul-13}, \cite{ref:PanMis-16}, \cite{ref:MoKimBra-12}, \cite{ref:CarAminSas-08}. 

Against the backdrop of the possibility of \eqref{e:linsys} undergoing structural changes of the types \ref{alter:states}-\ref{alter:inputs-to-states}-\ref{alter:inputs}, it is natural to ask whether certain fundamental system-theoretic properties of the altered system are preserved. For instance: What is the structure of \(B\) such that if a certain number of actuators in \eqref{e:linsys} fail, then the resulting system is still controllable? Moving one step further, is it possible to identify conditions on \(A\) together with a class of structural change of type \ref{alter:states} such that the altered system
\begin{equation}
	\label{e:perturbed linsys}
	\dot{\st}(t)=\sys' \st(t) + \inp\cont(t),
\end{equation}
is still controllable? Such questions, in general, turn out to be difficult from a complexity-theoretic standpoint \cite{ref:LiuPeqSinKar-16}, typically requiring a large number of computations (depending on the size of the system) to be performed to arrive at the answers.

In this article we study the simplest of such problems:
\begin{equation}
	\label{e:key problem}
	\left\{
	\begin{aligned}
		& \text{Given \(A\), identify a suitably ``minimal'' \(B\) such that}\\
		& \text{\eqref{e:linsys} is controllable even after the connection}\\
		& \text{between any two system states is severed}.
	\end{aligned}
	\right.
	\tag{\(\mathcal P\)}
\end{equation}
The problem \eqref{e:key problem} may look deceptively simple, but it is a computationally difficult problem: indeed, it can be recast in terms of a well-known hard combinatorial problem --- see Remarks \ref{r:2.1}, \ref{r:2.2}, and \ref{r:depend} for details.

We will resort to structural systems theory to solve \eqref{e:key problem}. Structural systems theory deals with system-theoretic properties that depend on the sparsity pattern of the interconnections between the system states and control inputs. More precisely, the locations of zeroes in the system and control matrices of \eqref{e:linsys} provide crucial information about controllability and other system-theoretic properties. This approach turns out to be very useful in the context of \eqref{e:key problem} since it is naturally fine-tuned to observing whether the entries of \(\sys\) and \(\inp\) are zeros.

The literature on structural system theory is comprehensive: The key concepts have been explored in several articles, e.g., \cite{ref:LiuSloBar-11, ref:PeqKarAgu-13,ref:PeqKarAgu-131, ref:PeqKarAgu-15, ref:Ols-15, ref:PeqKarAgu-16, ref:PeqKarAgui-16}. Over the past three decades, several verification results have been proposed for \eqref{e:key problem}; e.g., in \cite{ref:Rec-90} the authors explored the impact of directed link failures on structural controllability of the system, \cite{ref:JruDru-13} analysed the connection between controllability and standard network parameters such as topological transitivity and degree, etc.
Our problem \eqref{e:key problem} is closely related to the problems treated in \cite{ref:JafAjoAgh-11} and \cite{ref:LiuMoPeqSinKar-13}, our work differs from others in the sense that we do not restrict to a special class of systems but are interested in solving \eqref{e:key problem} a general class of systems by using an efficient system-synthesis mechanism; see Remark \eqref{r:2.3} for a detailed discussion. To the best of our knowledge, our approach to solve \eqref{e:key problem} is novel, and the advantage of our system-synthesis mechanism lies in that it can be effectively adapted to solve other similar problems. We provide illustrations of our approach on standard benchmark IEEE power networks to establish its effectiveness.

The rest of the article is organised as follows: \S\ref{s:background} reviews certain concepts and results from discrete mathematics that will be used in this article. The precise problem statement and our main results are presented in \S\ref{s:main results}. An illustrative example is presented in \S\ref{s:example}. We conclude with a summary of this article and a set of future directions in \S\ref{s:conclusion}.

\section{Background}
\label{s:background}
Structural system theory starts with the representation of \eqref{e:linsys} as a directed graph \(\GAB\): Let  \(\VA=\lbrace v_1, v_2, \ldots,  v_d\rbrace\), and \(\VB=\lbrace r_1, r_2, \ldots, r_m\rbrace\) be the state vertices and control/input vertices corresponding to the states \(\st(t)\in \R^d\) and the control \(\cont(t) \in \R^m\) of the system \eqref{e:linsys}. Similarly, let \( \EA=\lbrace (v_j, v_i)| a_{ij}\neq 0\rbrace\) and \(\EB=\lbrace (r_j, u_i)| b_{ij}\neq 0\rbrace\) where \(a_{ij}\) and \(b_{ij}\) are the elements of the matrix \(\sys\) and \(\inp\). The directed graph \(\GAB\) is represented as \((V, E)\), where \(V=\VA \sqcup \VB\) and  \(E= \EA \sqcup \EB\) where \(\sqcup\) represents the disjoint union. In \(\GAB\), \(\EA\) symbolizes the set of edges between the state vertices, and  \(\EB\)  symbolizes the set of edges from the control vertices to the state vertices. In the similar manner, we can define \(\GA= (\VA, \EA)\) as a directed graph which considers the state vertices and the edges in between them.

A directed graph \(G_{s} = (V_{s}, E_{s})\) with \(V_{s} \subset \VA\) and \(E_{s} \subset \EA\) is called a \emph{subgraph} of \(\GA\). When \(U \subset V(\GA)\), the \emph{induced subgraph} \(G[U]\) consists of \(U\) and all the edges whose endpoints are contained in \(U\). A sequence of edges \(\lbrace(v_1, v_2), (v_2, v_3), \ldots, (v_{k-1}, v_k)\rbrace\), where each \((v_i, v_j)\in \EA\) with all the vertices distinct, is called a \emph{directed path} from \(v_1\) to \(v_k\). The directed graph \(G(A)\) is said to be \emph{strongly connected} if there exists a directed path between every pair of vertices in it.  A \emph{strongly connected component (SCC)} of \(\GA\) is a maximal subgraph such that for every \(v\), \(w\) \(\in\) \(\VA\) in the subgraph there exist a directed path from \(v\) to \(w\) and a directed path from \(w\) to \(v\). A \emph{Directed acyclic graph} (DAG) of \(\GA\) represents each SCC as a node and a directed edge exists between two nodes iff there exists a directed edge connecting the correponding SCCs in \(\GA\). The DAG associated with \(\GA\) can be efficiently generated in \(\mathcal{O}(\abs{\VA}+\abs{\EA})\).   For a set \(X \subset \VA\), the set of edges of \(\GA\) entering \(X\) is termed as \emph{incut} of \(X\). Similarly the set of edges of \(\GA\) leaving \(X\) is called \emph{outcut} of \(X\). The size of the incut and outcut associated with \(X\) are denoted by \(d^{+}(X)\) and \(d^{-}(X)\), and is termed as \emph{in-degree} and \emph{out-degree} of \(X\) denoted by \(d^{+}(X)\) and \(d^{-}(X)\). \(\Delta^{+}(\GA)\) and \(\Delta^{-}(\GA)\) represents the maximum out-degree and in-degree of \(\GA\). Let \(v \in \VA\), the in-degree \(d^{-}(v)\) is the number of edges terminating at \(v\). The out-degree \(d^{+}(v)\) is the number of edges leaving \(v\). 

 The directed graph \(\GA\) can be represented by a undirected bipartite graph in the following standard fashion: \(\HA\Let(\bil \sqcup \bir, \Gamma)\) where \(\bil\Let{\lbrace v_1^1, v_2^1, \ldots, v_d^1 \rbrace} \), \(\bir\Let{\lbrace v_1^2, v_2^2, \ldots, v_d^2 \rbrace} \), and \( \Gamma = \lbrace (v_j^1, v_i^2)| a_{ij}\neq 0\rbrace \) as portrayed in Fig. 1. 
 \tikzset{middlearrow/.style={
        decoration={markings,
            mark= at position 0.5 with {\arrow{#1}} ,
        },
        postaction={decorate}
    }
}
\begin{figure}[th]
\label{fig:2.1}
\begin{center}
\begin{tikzpicture}
\node[shape=circle,draw=black] (1) at (0,0) {$v_1$};
    \node[shape=circle,draw=black] (2) at (-1.5,-2) {$v_2$};
    \node[shape=circle,draw=black,] (3) at (1.5,-2) {$v_3$};
     \path [->,solid,draw=blue](1) edge node[right] {} (2); 
     \path [->,solid,draw=blue](3) edge node[right] {} (1); 
     \path [->,solid,draw=blue](2) edge node[right] {} (3); 
   \end{tikzpicture}
\qquad 
\begin{tikzpicture}
\node[shape=circle,draw=black] (1) at (5,0) {$v_1^1$};
    \node[shape=circle,draw=black] (2) at (5,-1) {$v_2^1$};
    \node[shape=circle,draw=black,] (3) at (5,-2) {$v_3^1$};
    \node[shape=circle,draw=black] (1') at (7,0) {$v_1^2$};
    \node[shape=circle,draw=black] (2') at (7,-1) {$v_2^2$};
    \node[shape=circle,draw=black] (3') at (7,-2) {$v_3^2$};
    \draw (1) edge[-,draw=blue] (2') ;
    \draw (2) edge[-,draw=blue] (3') ;
      \draw (3) edge[-,draw=blue] (1') ;
\end{tikzpicture}
\end{center}
\caption{A digraph and its bipartite representation.}
\end{figure}
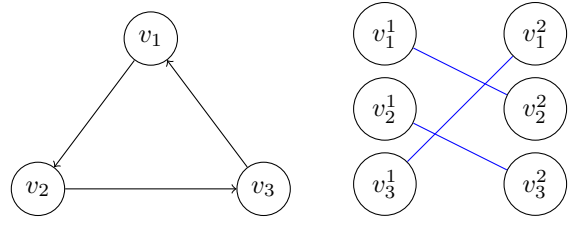

A \emph{matching} \(\match\) in \(\HA\) is a subset of edges in \(\Gamma\) that do not share vertices. A \emph{maximum matching} \(\match\) in \(\HA\) is defined as a matching \(\match\) that has largest number of edges among all possible matchings. A vertex is said to be \emph{matched} if it belongs to an edge in the matching \(\match\); otherwise, it is \emph{unmatched.} A matching \(\match\) in \(\HA\) is said to be \emph{perfect} if all the vertices in \(\HA\) are matched. A maximum matching \(\match\) can be found in \(\HA\) in polynomial time \(\mathcal{O}(\sqrt{\abs{\bil \sqcup \bir}}\abs{\Gamma})\). Note that a maximum matching \(\match\) may not be unique.
\begin{remark}
\label{r:2.1}
The problem of finding the family of all maximum matchings in a digraph is a sharp P-complete problem  \cite{ref:LovPlu-86}, which is extremely difficult to solve.
 \end{remark}
 \begin{remark}
 \label{r:2.2}
 We note that \eqref{e:key problem} can be transformed into a set-cover problem \cite{ref:LiuMoPeqSinKar-13,ref:LiuPeqSinKar-16}, which is known to be NP-hard . This clearly depicts the complexity associated with solving \eqref{e:key problem}.
\end{remark}
\begin{remark}
\label{r:2.3} The authors of \cite{ref:JafAjoAgh-11} approached the problem of establishing structural controllability under failures by assuming \(\sys\) to be a matrix with all the diagonal elements non-zero. Real-world networks, e.g., power networks do not satisfy this condition. The authors of \cite{ref:LiuMoPeqSinKar-13} considered the problem of structural observability when the system matrix \(\sys\) is irreducible and symmetric. This is a very restrictive class of systems for which analysing \eqref{e:key problem} is not difficult.  We do not restrict ourselves to such limiting assumptions. 
\end{remark}

For the directed graph \(\GAB\) corresponding to the system \eqref{e:linsys}
\begin{itemize}[leftmargin=*]
\item A vertex \(v\) is said to be \emph{accessible from the control vertices} if there exists a directed path terminating at \(v\) starting from atleast one of the control vertices; otherwise it is  said to be \emph{inaccessible from the control vertices}.
\item For a subset \(S \subset \VA\), the \emph{neighbourhood} of \(S\) is the set  \(\lbrace \NS(S)=v_j | (v_j, v_i) \in \EA\sqcup E_B, v_i \in S \rbrace \). Each vertex in \(\NS(S)\) is termed as a \emph{in-neighbour} of \(S\). The directed graph \(\GAB\) is said to have a \emph{dilation} if  there exists a set \(S \subset \VA\) such that \(\abs{\NS(S)} < \abs{S}\). 
\end{itemize}
A fundamental connection between the system theoretic property of structural controllability and certain structural properties of \(\GAB\) is given by
\begin{theorem}[\cite{ref:Lin-74}] 
\label{t:2.4}
The following are equivalent:
\begin{enumerate}[label={\rm (\alph*)}, widest=b, leftmargin=*, align=left]
\item The pair \((\sys,\inp)\) is structurally controllable.
\item The directed graph \(\GAB\) derived from \eqref{e:linsys} as described in \S\ref{s:background} has all the state vertices accessible from the control vertices, and \(\GAB\) has no dilation.
\end{enumerate}
\end{theorem}  

\noindent\textbf{Definition 1:} Given a \(\delta \in \lbrace 0,1\rbrace^d\), let \(\rt\Let\lbrace v_i: v_i\in V(\GA)\) \text{and} \(\delta_i=1\rbrace\) then, \(\rt\) is termed the \emph{root set} and \(v_i\in \rt\) is called a \emph{root vertex}.  Those directed edges terminating at one of the root vertices are termed as \emph{root edges}.
 
The following more recent structural result will also be needed in the sequel:
\begin{theorem}[\cite{ref:PeqKarAgu-13}]
\label{t:2.5}
Let \(\GA\) be the state digraph and \(\HA\) be the associated bipartite graph, then the following statements are equivalent:
\begin{enumerate}[label={\rm (\alph*)}, widest=b, leftmargin=*, align=left]
\item The set \(\rt\) is a dedicated input configuration.\footnote{ In this article \({\inp}\) is considered as \({\inp}\Let \diag(\delta)\), where \(\delta \in \lbrace 0,1\rbrace^d\). Here, if \(\delta_i=1\) state vertex \(v_i\) receives an input, while if \(\delta_i=0,\) it receives  no input.}
\item There exists a subset \(\mathcal{U}(\match)\subset \rt\) corresponding to the set of unmatched vertices of some maximum matching \(\match\) of \(\HA\), and a subset \(\mathcal{A} \subset \rt\) consisting of one state variable from SCC of \(\GA\) that has no incoming edge.
\end{enumerate}
\end{theorem}
A root set \(\rt\) obtained via Theorem \ref{t:2.5} ensures structural controllability of \(\GA\).

\section{Main Results}
\label{s:main results}
We catalogue some important notions specific to digraphs:

\noindent\textbf{Definition 2:} The minimum number of non-root edges whose removal makes \(G\) not structurally controllable is referred as the \emph{edge-controllability index} of the digraph \(G\), and is denoted by \(\ec_{R}(G)\) w.r.t root set \(R\). The digraph \(G\) is said to be \(k\)-edge-controllable if its edge-controllability is atleast \(k\).


 
Let the assume that a digraph \(G=(V, E)\) is structurally controllable w.r.t \(\rt\) then we introduce the notion of \emph{critical} edges in the digraph.

\noindent\textbf{Definition 3:}
 Given a digraph \(\gra=(V, E)\) and a  root set \(\rt\), an edge \((v_i,v_j) \in E\) is said to be \textit{critical/sensitive} if the digraph \(G_{ji}\) obtained by deleting edge \((v_i,v_j)\) is not structurally controllable with respect to \(\rt\).
  
For each edge \(e\in E(G)\), the criticality of \(e\) is examined by analysing structural controllability of digraph obtained by deleting \(e\) from \(G\). Therefore the problem of finding critical edges can be acheived in polynomial time as it is equivalent to analysing structural controllability of \(\abs{E}\) digraphs \cite{ref:KMur-87}. In view of Theorem \ref{t:2.4}, deletion of a critical edge creates either input inaccessibility or dilation or both with respect to \(\rt\). In other words, removal of a critical edge  \(e\) creates either of the two sets:
\begin{enumerate}[label={\rm (\alph*)}, widest=b, leftmargin=*, align=left]
\item An SCC \(\Xe\) with in-degree \(0\), i.e., \(\de(\Xe)=0\). As \(\Xe \cap \rt = \emptyset\), every  vertex \(v \in \Xe\) is inaccessible from \(\rt\).
\item A minimal set \(\Se\) such that \(\abs{\Se} > \abs{\NS(\Se)}\) and \(\Se \cap \rt = \emptyset\). 
\end{enumerate}
Figure (2) displays the critical edges, and sets \(\Xe\) and \(\Se\) corresponding to them, in an illustrative example.  
 \tikzset{middlearrow/.style={
        decoration={markings,
            mark= at position 0.5 with {\arrow{#1}} ,
        },
        postaction={decorate}
    }
}
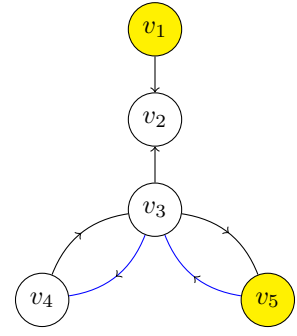
\begin{figure}
\label{f:example}
 \begin{center}

\begin{tikzpicture}
\node[shape=circle,draw=black,fill=yellow] (1) at (0,0) {$v_1$};
    \node[shape=circle,draw=black] (2) at (0,-1.2) {$v_2$};
    \node[shape=circle,draw=black] (3) at (0,-2.4) {$v_3$};
    \node[shape=circle,draw=black] (4) at (-1.5,-3.6) {$v_4$};
    \node[shape=circle,draw=black,fill=yellow] (5) at (1.5,-3.6) {$v_5$};
   
    \path [->,solid, draw=blue](1) edge node[right] {} (2);
     \path [->,solid](3) edge node[right] {} (2); 
   \draw[bend left,middlearrow={>},draw=blue]   (3) to node [auto] {} (4);
    \draw[bend left,middlearrow={>}]   (3) to node [auto] {} (5);
    \draw[bend left,middlearrow={>}]   (4) to node [auto] {} (3);
    \draw[bend left,draw=blue,middlearrow={>}]   (5) to node [auto] {} (3);
   
    \end{tikzpicture}

    \end{center}
    \caption{The yellow state vertices depicts the initial root set \(\rt=\lbrace v_1, v_5 \rbrace\). The critical edges are  \(e_1=(v_1, v_2), e_2=(v_5, v_3),\) and \(e_3=(v_3, v_4)\) coloured in blue.  Removal of edges  \(e_1, e_2\), and \( e_3\) yield sets \(S_{e_1}=\lbrace v_2, v_4 \rbrace\),
 \(X_{e_2}=\lbrace v_3, v_4\rbrace\), and \(X_{e_3}=\lbrace v_4 \rbrace\).}
    \end{figure}
    
   \begin{remark}
   \label{r:addedge}
   If \(G\) is \(2\)-edge controllable w.r.t \(\rt\), then \(G'\) obtained by adding an edge between any two vertices is also \(2\)-edge controllable w.r.t \(\rt\).
   \end{remark}

Let \(\bar{\sys}\) and \(\bar{\inp}\) represent the structured/sparsity pattern of \(\sys\) and \(\inp\). \footnote{Location of zeroes and non-zeroes in the system and control matrices of \eqref{e:linsys} of \(\sys\), \(\inp\).}.

 Now \eqref{e:key problem} can be reformulated in the following manner: Given a state digraph \(\gra(\bar{\sys})=(V_{\bar{\sys}}, E_{\bar{\sys}})\) corresponding to \(\bar{\sys}\), select a minimal root set \(\rtn\subset V(\gra(\bar{\sys}))\) such that \(\ec_{\rtn}(G(\bar{\sys})) \geq 2\).

Although \eqref{e:key problem} is a challenging problem as evidenced in Remark \ref{r:2.2}, we adopt a step-wise procedure to address it: 

\textbf{Step 1:} Suppose \(\gra (\bar{\sys})\) is structurally controllable w.r.t to an initial root set \(R\). Recall that an initial root set \(R\) can be computed by using Theorem \ref{t:2.5}. The first step involves obtaining subgraphs \(G_1, G_2, \ldots, G_k\) of \(\gra (\bar{\sys})\) such that each subgraph is structurally controllable w.r.t \(R\). The aim is to identify the additional root vertices such that each subgraph \(G_i\) has \(\ec_{R}(G_i) \geq 2\) for \(i= 1,2,\ldots, k\). This is accomplished by adopting an efficient algorithm/procedure to obtain the sets \(\Xe\)s (as in (a) above) along with  its critical edges.
 \begin{algo}
   \label{algo 1}
        \mbox{}
            \begin{description}
            \item[1] \textbf{Input:} Given \(G_1, G_2, \ldots, G_k\) be the subgraphs of \(\gra (\bar{\sys})\) such that each subgraph \(G_i\) has \(\Delta^{+}(G_i) \leq 2\) and is structurally controllable w.r.t \(R\) 
            \item[2]  \textbf{for} \(i\) from \(1\) to \(k\) \textbf{do}
            \item[2a] Identify critical edges of \(G_i\).
            \item[2b] Identify the SCC \(\Xe\) corresponding to each critical edge \(e\) of \(G_i\).
            \item[3]\textbf{end do}
            \item[4]\textbf{Output:} critical edges and \(\Xe\)s corresponding to them.
         \end{description}
        \end{algo} 
Algorithm \ref{algo 1} allows us to identify the \(\Xe\)s corresponding to the critical edges in the subgraph obtained from \(G(\bar{\sys})\).\footnote{Note that invoking Lemma \ref{t:SCC} allows us to consider only \(\Xe\)s to attain resilience w.r.t edge failure, as opposed to both \(\Xe\)s and \(\Se\)s.}
\begin{proposition}
\label{p:algocom}
Algorithm \ref{algo 1} has polynomial-time complexity.
\end{proposition}
A proof of Proposition \ref{p:algocom} is presented in \S\ref{s:Appendix}.
\begin{remark}
\label{r:algoinf}
It is  not necessary that total number of vertices in all the subgraphs encountered in Algo \ref{algo 1} is equal to the number of  vertices in \(G\).
\end{remark}
Consider all the \(\Xe\)s obtained by removal of edges, say \({\lbrace \Xe^i\rbrace}_{i \in \I}\) with \(\I=\lbrace 1,2,3,\ldots, n\rbrace\) and \(\Ze=\lbrace v_1,v_2,\ldots v_d\rbrace\). Define the family of subsets \(\F_j=\lbrace i\in \I: v_j \in \Xe^i \rbrace\) for \(j=1,2,\ldots, d\). If there exists an \(L \subset \lbrace 1,2,\ldots d\rbrace\) such that
\[
\I \subset \bigcup_{j \in L}\F_j
\]
then the family \({\lbrace \F_j \rbrace}_{j \in L}\) covers \(\I\) and
\begin{align*}
\rtn&= \bigcup_{j \in L}\lbrace v_j \rbrace \cup \rt,
\end{align*}
where \(\rtn\) is the minimal root set such that each subgraph is robust w.r.t one edge failure. There exists a greedy approximation algorithm that provides a set cover that is \(\ln n\) larger than an optimal set cover \cite{ref:CorLeiRivSte-09}.

The second step involves addition of vertices to the subgraphs to retrieve the original \(G(\bar{\sys}) \) such that \(G(\bar{\sys})\) is also resilient to an edge failure w.r.t \(\rtn\).

\textbf{Step 2:}
 To obtain the original graph \(G(\bar{A})\), vertices are added to the subgraphs. Let \(G_1\) be a subgraph obtained from step \(1\) above. Suppose \(z\) is added to \(G_1\).\footnote{Note that addition of \(z\) can connect subgraphs also. Here its is assumed that the edges corresponding to \(z\) are added to only one subgraph i.e. \(G_1\). The addition of \(z\) to more than one subgraph is examined in section \S\ref{s:example}.} Let \(G_1^{*}\) be the new subgraph obtained by adding \(z\) and its edges to \(G_1\).
 
 If \(S \subset V(G_1)\) be a set satisfying the following conditions:
 \begin{enumerate}
 \item \(\abs{S}=\abs{N^{-}(S)}\), and
 \item even after removal of any one edge \(e=(x,y)\), where \(x \in N^{-}(S)\) and \(y \in S\), condition 1) holds,
 \end{enumerate}
then \(S\) is termed as a \emph{critical set}. The following theorem ensures that \(G_1^{*}\) obtained by adding a vertex \(z\) to the subgraph \(G_1\) is also robust w.r.t a failure with the same root set \(\rtn\).
\begin{theorem}
\label{l:addnode}
Let \(z\) and its edges are added to a subgraph \(G_1\) and  \(S \subset V(G_1)\) be a critical set.
If  \(G_1\) is \(2\)-edge controllable w.r.t \(\rtn\), then \(G_1^{*}=G_1 \cup \lbrace z \rbrace\) obtained by adding a new vertex \(z\) is also \(2\)-edge controllable w.r.t \(\rtn\), if the following conditions are satisfied:
\begin{enumerate}[label={\rm (\alph*)}, widest=b, leftmargin=*, align=left]
\item \(z\) has at least two in-neighbours from \(G_1\), and
\item if \(N^{-}(z) \cap N^{-}(S)\neq \emptyset\), then \(z\) has at least two in-neighbours not contained in \(N^{-}(S)\).  
\end{enumerate}
\end{theorem}
We provide a proof in \S\ref{s:Appendix}.
Theorem \eqref{l:addnode} allows us to add vertices to the graph consecutively such that robustness of graph w.r.t failure is preserved. This completes the two-step procedure involved to obtain the root set \(\rtn\) for \(G(\bar{\sys})\) such that \(\ec_{\rtn}(G) \geq 2\). 
\begin{remark}
\label{r:depend}
The minimal root set \(\tilde{R}\) obtained in step \(1\) depends on the initial root set \(\rt\) and the way the subgraphs have been constructed there from the original graph \(\gra(\bar{\sys})\).
\end{remark}
\section{Illustrative Example}
\label{s:example}
Let us consider the network topology \(G\) of the IEEE 14-bus system \cite{ref:smagrid} depicted in Fig. (3). \textbf{Each undirected edge between the two vertices denotes bidirectional edges between them.} Let \(R=\lbrace v_8, v_{10}\rbrace\) be the initial root set such that \(G\) is structurally controllable w.r.t \(R\). Two subgraphs  \(G_1= \lbrace v_1, v_2, v_3, v_4, v_7, v_8\rbrace\) and \(G_2=\lbrace v_6, v_{10}, v_{11}, v_{12}, v_{13}, v_{14} \rbrace\) are obtained such that \(G_1 \sqcup G_2\) is structurally controllable w.r.t \(R\) as shown in Fig (4). The critical edges corresponding to \(\rt\) are:
\(e_1=(v_2,v_1)\), \(e_2=(v_3,v_2)\), \(e_3=(v_4,v_3)\), \(e_4=(v_7,v_4)\), \(e_5=(v_8,v_7)\), \(e_6=(v_{10},v_{11})\), \(e_7=(v_{11},v_6)\), \(e_8=(v_{6},v_{12})\), \(e_9=(v_{12},v_{13})\), \(e_{10}=(v_{13},v_{14})\). Each critical edge creates an \(\Xe\) corresponding to it. \(X_{e_1}=\lbrace v_1 \rbrace\), \( X_{e_2}=\lbrace v_1, v_2 \rbrace\), \(X_{e_3}=\lbrace v_1,v_2,v_3 \rbrace\), \(X_{e_4}=\lbrace v_1, v_2, v_3, v_4 \rbrace\), \(X_{e_5}=\lbrace  v_1, v_2, v_3, v_4, v_7 \rbrace\), \( X_{e_6}=\lbrace v_{11},v_6, v_{12},v_{13},v_{14} \rbrace\), \(X_{e_7}=\lbrace v_6, v_{12},v_{13},v_{14} \rbrace\), \(X_{e_8}=\lbrace v_{12},v_{13},v_{14} \rbrace\), \(X_{e_9}=\lbrace v_{13},v_{14} \rbrace\), and \(X_{e_{10}}=\lbrace v_{14} \rbrace\). By using greedy approximation algorithm for set cover as in step \(1\), the additional root vertices required for ensuring robustness w.r.t an arbitrary edge failure is computed. The solution of the set cover is \(\rtn=\lbrace v_1, v_8, v_{10}, v_{14}\rbrace\) displayed in Fig. (5). Hence, \(\rtn\) is a root set such that each subgraph is \(2\)-edge controllable. 

Now \(\lbrace v_9\rbrace\) is added to the subgraph \(G_2\) with its undirected edges \((v_{9},v_{10})\) and \((v_{9},v_{14})\) to obtain \(G_2^{*}\) as shown in Fig. (6). As each undirected edge between the two vertices represent two bidirectional edges between them, four edges correponding to \(v_9\) are added. Since \(v_9\) has two in-neighbours \( v_ {10}\) and \(v_{14}\) condition(a) of Theorem \ref{l:addnode} is satisfied . Also \(G_2\) has no critical set in it. Hence, \(\ec_{\rtn}(G_2^{*}) \geq 2\). Similarly, \(\lbrace v_5 \rbrace\) and its edges \((v_5,v_1)\) and \((v_5,v_6)\) is added which connects  \(G_1\) and \(G_2^{*}\) as shown in Fig  (7). Let the new graph obtained by adding \(v_5\) be \(G^{*}\). \(G^{*}\) has no critical set and \(v_5\) has two in-neighbours \(v_1\) and \(v_6\). As both the conditions of Theorem \ref{l:addnode} are satisfied, \(\ec_{\rtn}(G^{*}) \geq 2\). The remaining edges corresponding to \(v_5\) and \(v_9\) are added to retrieve the original graph \(G\) of IEEE 14-bus power system depicted in Fig. (8). In view of Remark \ref{r:addedge}, \(\ec_{\rtn}(G) \geq 2\). Therefore, \(G\) is resilient w.r.t to an arbitrary edge failure with \(\rtn=\lbrace v_1, v_8, v_{10}, v_{14}\rbrace\). 
\begin{figure}[h!]
\begin{center}
   \begin{tikzpicture}
\node (1) at (0,0) [] {$v_1$};
\node (2) at (1.5,0) [] {$v_2$};
\node (3) at (3,0) [] {$ v_3 $};
\node  (4) at (3,-1) [] {$ v_4 $};
\node (5) at (0,-1) [] {$v_5$};
\node (6) at (0,-2) [] {$v_6$};
\node (7) at (4.5,-1) [] {$ v_7 $};
\node (8) at (4.5,0) [] {$ v_8 $};
\node (9) at (3,-2) [] {$v_9$};
\node (10) at (3,-4.5) [] {$v_{10}$};
\node (11) at (0,-3) [] {$ v_{11} $};
\node (12) at (1.5,-3) [] {$ v_{12} $};
\node (13) at (3.7,-3) [] {$v_{13}$};
\node (14) at (5.5,-3) [] {$v_{14}$};
\draw (1) edge[-] (2) (2) edge[-] (3)  (1) edge[-] (5) 
 (2) edge[-] (5)  (2) edge[-] (4)  (3) edge[-] (4)
 (4) edge[-] (7)  (7) edge[-] (8)  (4) edge[-] (9)
 (5) edge[-] (6)  (6) edge[-] (11)  (6) edge[-] (12)
  (12) edge[-] (13)  (7) edge[-] (9) (9) edge[-] (10)
   (9) edge[-] (14)  (10) edge[-] (11)  (13) edge[-] (14)
    (6) edge[-] (13)  (4) edge[-] (5);
\end{tikzpicture}
 \end{center}
\caption{Network topology of IEEE-14 bus power system \(G\).}
\end{figure}
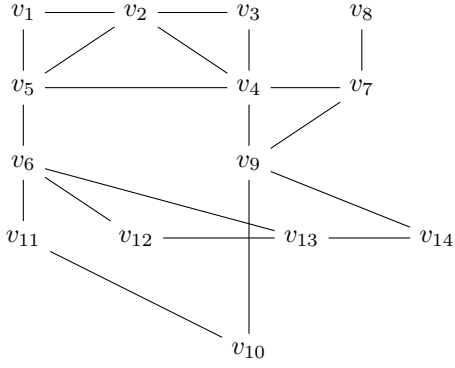

 \begin{figure}[h!]
 \begin{center}
  \begin{tikzpicture}
\node[] (1) at (0,0) [] {$v_1$};
\node (2) at (1.5,0) [] {$v_2$};
\node (3) at (3,0) [] {$ v_3 $};
\node  (4) at (3,-1) [] {$ v_4 $};
\node (6) at (0,-2) [] {$v_6$};
\node (7) at (4.5,-1) [] {$ v_7 $};
\node (8)[shape=circle,draw=brown] at (4.5,0) [] {$ v_8 $};
\node (10)[shape=circle,draw=brown] at (3,-4.5) [] {$v_{10}$};
\node (11) at (0,-3) [] {$ v_{11} $};
\node (12) at (1.5,-3) [] {$ v_{12} $};
\node (13) at (3.7,-3) [] {$v_{13}$};
\node (14)[] at (5.5,-3) [] {$v_{14}$};
\draw (1) edge[-] (2) (2) edge[-] (3)  (6) edge[-] (11)  (6) edge[-] (12)  
   (7) edge[-] (8)  (12) edge[-] (13) (3) edge[-] (4)
 (4) edge[-] (7) 
   (10) edge[-] (11)  (13) edge[-] (14);
\end{tikzpicture}
\end{center}
\caption{The brown colored vertices \(\lbrace v_8,v_{10}\rbrace\) represent the initial root set  \(\rt\). \(G_1= \lbrace v_1, v_2, v_3, v_4, v_7, v_8\rbrace\) and \(G_2=\lbrace v_6, v_{10}, v_{11}, v_{12}, v_{13}, v_{14} \rbrace\) are the subgraphs.}
 \end{figure}
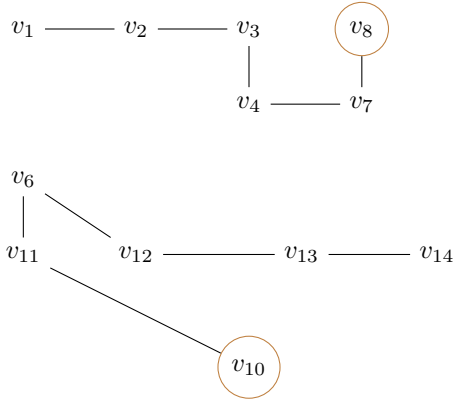
  \begin{figure}[h!]
\begin{center} 
  \begin{tikzpicture}
\node[shape=circle,draw=gray,dashed] (1) at (0,0) [] {$v_1$};
\node (2) at (1.5,0) [] {$v_2$};
\node (3)[] at (3,0) [] {$ v_3 $};
\node  (4) at (3,-1) [] {$ v_4 $};
\node (6) at (0,-2) [] {$v_6$};
\node (7)[] at (4.5,-1) [] {$ v_7 $};
\node (8)[shape=circle,draw=brown] at (4.5,0) [] {$ v_8 $};
\node (10)[shape=circle,draw=brown] at (3,-4.5) [] {$v_{10}$};
\node (11)[] at (0,-3) [] {$ v_{11} $};
\node (12)[]at (1.5,-3) [] {$ v_{12} $};
\node (13) at (3.7,-3) [] {$v_{13}$};
\node (14)[shape=circle,draw=gray,dashed] at (5.5,-3) [] {$v_{14}$};
\draw (1) edge[-] (2) (2) edge[-] (3) (3) edge[-] (4)
 (4) edge[-] (7)  (6) edge[-] (11)  (6) edge[-] (12)
   (7) edge[-] (8)  (12) edge[-] (13)  
   (10) edge[-] (11)  (13) edge[-] (14);
\end{tikzpicture}
 \end{center}
\caption{The dashed gray vertices \(\lbrace v_1, v_{14}\rbrace\) represent the vertices where additional inputs are applied.}
 \end{figure}
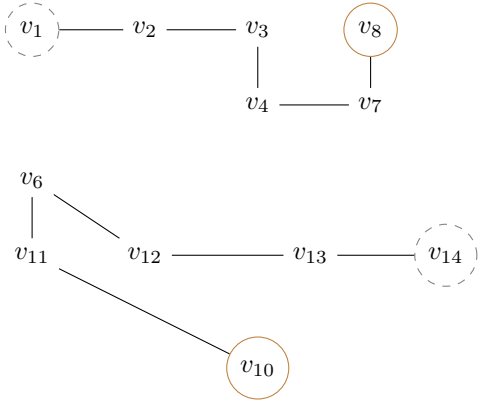
 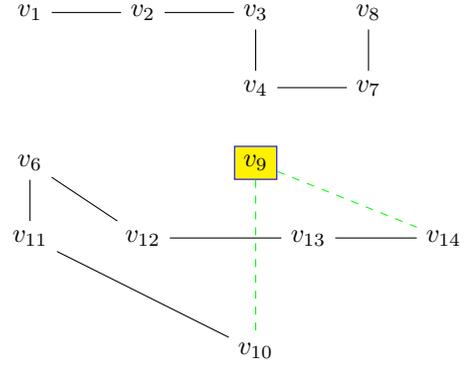
\begin{figure}[h!]
 \begin{center}
\begin{tikzpicture}
\node (1) at (0,0) [] {$v_1$};
\node (2) at (1.5,0) [] {$v_2$};
\node (3) at (3,0) [] {$ v_3 $};
\node[]  (4) at (3,-1) [] {$ v_4 $};
\node[] (6) at (0,-2) [] {$v_6$};
\node (7) at (4.5,-1) [] {$ v_7 $};
\node (8) at (4.5,0) [] {$ v_8 $};
\node[draw=blue,fill=yellow] (9) at (3,-2) [] {$v_9$};
\node (10) at (3,-4.5) [] {$v_{10}$};
\node (11) at (0,-3) [] {$ v_{11} $};
\node (12) at (1.5,-3) [] {$ v_{12} $};
\node (13) at (3.7,-3) [] {$v_{13}$};
\node (14) at (5.5,-3) [] {$v_{14}$};
\draw (1) edge[-] (2) (2) edge[-] (3)  
   (7) edge[-] (8)  (12) edge[-] (13) (9) edge[-,draw=green,dashed] (10)
   (9) edge[-,draw=green,dashed] (14) (3) edge[-] (4)
 (4) edge[-] (7) (6) edge[-] (11)  (6) edge[-] (12)
   (10) edge[-] (11)  (13) edge[-] (14);
\end{tikzpicture}
\end{center}
\caption{\(\lbrace v_9\rbrace\) is added to graph to \(G_2\)}
\end{figure}
 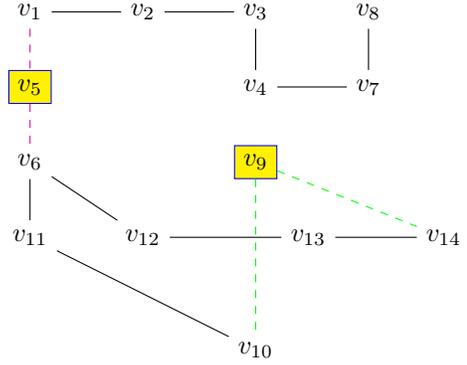
\begin{figure}[h!]
 \begin{center}
\begin{tikzpicture}
\node (1) at (0,0) [] {$v_1$};
\node (2) at (1.5,0) [] {$v_2$};
\node (3) at (3,0) [] {$ v_3 $};
\node[]  (4) at (3,-1) [] {$ v_4 $};
\node[draw=blue,fill=yellow] (5) at (0,-1) [] {$v_5$};
\node[] (6) at (0,-2) [] {$v_6$};
\node (7) at (4.5,-1) [] {$ v_7 $};
\node (8) at (4.5,0) [] {$ v_8 $};
\node[draw=blue,fill=yellow] (9) at (3,-2) [] {$v_9$};
\node (10) at (3,-4.5) [] {$v_{10}$};
\node (11) at (0,-3) [] {$ v_{11} $};
\node (12) at (1.5,-3) [] {$ v_{12} $};
\node (13) at (3.7,-3) [] {$v_{13}$};
\node (14) at (5.5,-3) [] {$v_{14}$};
\draw (1) edge[-] (2) (2) edge[-] (3)  
   (7) edge[-] (8)  (12) edge[-] (13) (9) edge[-,draw=green,dashed] (10)
   (9) edge[-,draw=green,dashed] (14) (3) edge[-] (4)  (1) edge[-,draw=magenta,dashed] (5) 
 (4) edge[-] (7) (6) edge[-] (11)  (6) edge[-] (12) (5) edge[-,draw=magenta,dashed] (6)
   (10) edge[-] (11)  (13) edge[-] (14);
\end{tikzpicture}
\end{center}
\caption{\(\lbrace v_5\rbrace\) is added to the subgraphs \(G_1\) and \(G_2^{*}\).}
\end{figure}
 \begin{figure}[h!]
 \begin{center}
\begin{tikzpicture}
\node (1) at (0,0) [] {$v_1$};
\node (2) at (1.5,0) [] {$v_2$};
\node (3) at (3,0) [] {$ v_3 $};
\node[]  (4) at (3,-1) [] {$ v_4 $};
\node[draw=blue,fill=yellow] (5) at (0,-1) [] {$v_5$};
\node[] (6) at (0,-2) [] {$v_6$};
\node (7) at (4.5,-1) [] {$ v_7 $};
\node (8) at (4.5,0) [] {$ v_8 $};
\node[draw=blue,fill=yellow] (9) at (3,-2) [] {$v_9$};
\node (10) at (3,-4.5) [] {$v_{10}$};
\node (11) at (0,-3) [] {$ v_{11} $};
\node (12) at (1.5,-3) [] {$ v_{12} $};
\node (13) at (3.7,-3) [] {$v_{13}$};
\node (14) at (5.5,-3) [] {$v_{14}$};
\draw (1) edge[-] (2) (2) edge[-] (3)  (1) edge[-] (5) 
 (2) edge[-,draw=blue,dashed] (5)  (2) edge[-,draw=blue,dashed] (4)  (3) edge[-] (4)
 (4) edge[-] (7)  (7) edge[-] (8)  (4) edge[-,draw=blue,dashed] (9)
 (5) edge[-] (6)  (6) edge[-] (11)  (6) edge[-] (12)
  (12) edge[-] (13)  (7) edge[-,draw=blue,dashed] (9) (9) edge[-] (10)
   (9) edge[-] (14)  (10) edge[-] (11)  (13) edge[-] (14)
    (6) edge[-,draw=blue,dashed] (13)  (4) edge[-,draw=blue,dashed] (5);
\end{tikzpicture}
\end{center}
\caption{The original graph is obtained at this stage; It is resilient to one edge failure w.r.t \(\rtn=\lbrace v_1, v_8, v_{10}, v_{14} \rbrace\).}
\end{figure}
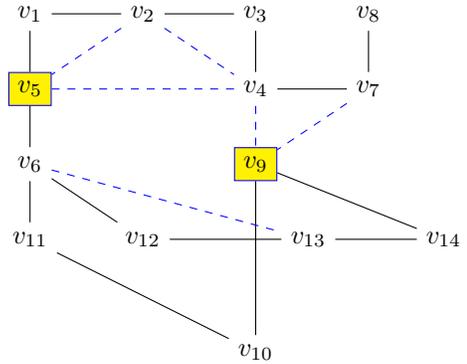

\section{Conclusions and future directions}
\label{s:conclusion}
In this article we presented the problem of finding the minimal number of actuators that ensure that the system remains structurally controllable under external malicious attack. The scenario considered here is the loss of physical connection between the two system states. Due to the combinatorial nature of this problem, we proposed an efficient system-synthesis mechanism to obtain robustness with respect to an edge failure. The effectiveness of our technique is illustrated by using a standard IEEE bus power system. In future we aim to study and develop this technique further to tackle interesting combinatorial problems. A natural extension would be to  determine the applicability of this technique to multiple edge failures occurring simultaneously.

\section{Appendix A}
\label{s:Appendix}
\begin{lemma}
\label{t:SCC}
 Given a digraph \(G=(V, E)\) structurally controllable w.r.t root set \(\rt\) such that \(\Delta^{+}(\gra)\leq 2\) then, for each nonempty \(S_{e_j}\), there exists a nonempty \(X_{e_i}\) where \(e_j, e_i \) \(\in E\) are the critical edges.
\end{lemma} 
\begin{proof} Consider a digraph \(G\) structurally controllable w.r.t \(\rt\).
Here we will show that if \(X_{e_i}= \emptyset\), then \(S_{e_j}= \emptyset\).

\noindent Suppose \(S_{e_j} \neq \emptyset\). We know that it is obtained by removal of a critical edge \(e_j=(w,y)\), where \(w, y \in V(G)\) and \(y \in S_{e_j}\). Note that the case \(\abs{S_{e_j}}=1\) is trivial as \(\NS(S_{e_j})=0\). Consider \(S_{e_j}\) such that \(\abs{S_{e_j}}\geq 2\).

 \noindent As \(S_{e_j}\) is a minimal set it satisfies the following criterion:
\begin{itemize}[leftmargin=*]
\item \(\abs{S_{e_j}}=\abs{\NS(S_{e_j})} + 1\)
\item  Every vertex in \(\NS(S_{e_j})\) is adjacent to at least two vertices in \(S_{e_j}\).
\end{itemize}
\noindent As \(\Delta^{+}(G) \leq 2\), every vertex in \(\NS(S_{e_j})\) is adjacent to exactly two vertices in \(S_{e_j}\).

\noindent \textit{Claim}: There exists a vertex \(v \in S_{e_j}\) with \(\de(v)=1\) and \(v \neq y\).
\noindent Firstly it is proved that there exists at least one vertex belonging to \(S_{e_j}\) whose in-degree is at most \(1\). Consider the induced subgraph of \(S_{e_j}\) and \(\NS(S_{e_j})\).

\noindent Suppose that every vertex in \(S_{e_j}\) has in-degree at least \(2\), then 
\begin{align*}
\sum_{v \in S_{e_j}}\de(v) \geq  2 \abs{S_{e_j}}
\end{align*} 
\begin{align*}
\sum_{v \in \NS(S_{e_j})}\dep(v) &= 2 \abs{\NS({S_{e_j}})}\\
 & = 2 (\abs{S_{e_j}}-1)\\  & < \sum_{v \in S_{e_j}} \de(v)
\end{align*}
which is a contradiction. Therefore there exists at least one vertex whose in-degree is at most \(1\). For any vertex \(v \in S_{e_j}\), if \(\de(v)= 0\) then it contradicts the minimality of \(S_{e_j}\). This confirms that there exist atleast one vertex, say \(z\) with \(\de(z)= 1\)  

\noindent Secondly we need to prove that \(z \neq y\).
This can also be proved using the same argument as above. Suppose that \(y \in S_{e_j}\) is the only vertex with in-degree \(1\) and rest of vertices in \(S_{e_j}\) has in-degree at least \(2\) in the induced subgraph of \(S_{e_j}\) and \(\NS(S_{e_j})\), then
\begin{align*}
\sum_{v \in S_{e_j}} \de(v) & \geq 1 + 2 (\abs{S_{e_j}}-1)\\ & \geq 1 + 2 \abs{\NS({S_{e_j}})}\\
 &\geq 1 + \sum_{v \in \NS({S_{e_j}})}\dep(v)
\end{align*}
\noindent which is a contradiction. This proves that there exists at least one vertex \(z \neq y \) in \(S_{e_j}\) s.t   \(\de(z)=1\). Therefore the edge \(e_i=(m, z)\) corresponding to \(z\), where \(m \in \NS(S_{e_j})\), is a critical edge as its removal makes \(z\) inaccessible from the root set \(\rt\). Therefore, \(X_{e_i}=\lbrace z \rbrace\) which is non-empty. Contradiction  
\end{proof}
\textit{Proof of Proposition \ref{p:algocom}}. For each subgraph \(G_i\), step \(2a\) deals with finding the critical edges that can be achieved in polynomial time as it is equivalent to analysing structural controllability of \(\abs{E(G_i)}\) digraphs.  It has  complexity of \(\mathcal{O}(\sqrt{V(G_i)}\abs{E(G_i)}^2)\). Step \(2c\) computes the SCCs corresponding to each critical edge in \(G_i\). The SCCs can be obtained by using DAG with complexity \(\mathcal{O}(\abs{V(G_i)}+\abs{E(G_i)})\). Similar steps are followed for each subgraph in the for loop. Hence,  Algorithm \ref{algo 1} has polynomial-time complexity.  

\textit{Proof of  Theorem \ref{l:addnode}}.
We know that every vertex in \(\gra_1\) is input accessible from the root set \(\rtn\). Condition (a) ensures that \(z\) added to the subgraph \(G_1\) is also accessible from  \(\rtn\). Since in-degree of \(z\) is at least \(2\), removal of an edge  terminating at \(z\) does not result in input-inaccessibility w.r.t \(\rtn\).
 
 We will prove condition (b) by analysing the following cases:
 \begin{enumerate}[leftmargin=*] 
 \item Suppose \(z\) has all its in-neighbours contained in \(\NS(S)\). By condition (a), \(z\) has at least two in-neighbours. Since \(S\) is a critical set, \(\abs{S}=\abs{N^{-}(S)}\). Addition of \(z\) creates a creates a new set \(S_{z}\) that satisfies the following criterion:
 \begin{itemize}
 \item \(\abs{S_{z}}= \abs{S} + 1\) as \(z\) is added to \(G_1\).
 \item \(\abs{\NS (S_{z})}= \abs{\NS(S)}\)
\end{itemize}     
This shows that \(\abs{\NS (S_{z})} = \abs{S_{z}}-1\). So, \(G_1^{*}\) has dilation in it. 
\item Suppose \(z\) has only one in-neighbour, say \(k \in V(G_1)\), not contained in \(\NS(S)\). Let \(e=(k,z)\). Suppose \(e\) is removed from \(G_1\). Then all the in-neighbours of \(z\) lie in \(\NS(S)\),  which again lead to dilation after removal of the edge \(e\) from \(G_1\).
\item Suppose \(z\) has at least two in-neighbours. Let \(k\) and \(m\) are the two in-neighbours of \(z\) not contained in \(\NS(S)\). Let \(e_1=(k,z)\) and \(e_2=(m,z)\). Addition of \(z\) creates a creates a new set \(S_{z}\) with \(\abs{\NS(S_z)}\geq \NS(S)+2\). Therefore,  removal of any one edge does not result in dilation. 
\end{enumerate} 
This proves that if \(z\) has at least two in-neighbours not contained in \(\NS(S)\), then it does not introduce dilation in \(G_1^{*}\) after removal of an edge.

\bibliographystyle{siam}
\bibliography{./ref}

\bigskip
\bigskip
\end{document}